%% file: main.tex
\newcommand{\longversion}[1]{#1}
\newcommand{\shortversion}[1]{}

\documentclass[a4paper,11pt]{article}
\input{preamble}
\title{Manipulating Node Similarity Measures in Networks}
\author{Palash Dey$^\star$ and Sourav Medya$^\dagger$\\
$^\star$Indian Institute of Technology, Kharagpur\\
$^\dagger$Northwestern University - Kellogg School of Management\\
$^\star$palash.dey@cse.iitkgp.ac.in, $^\dagger$sourav.medya@kellogg.northwestern.edu}

\begin{document}

\maketitle

\input{abstract}
\input{introduction}
\input{prelim}

\input{results}
\input{exp.tex}
\input{previous_work}
\input{conclusion}

\bibliographystyle{unsrt}
\bibliography{main}

\end{document}

%% file: preamble.tex
\usepackage{xspace}
\usepackage{color}
\definecolor{antiquebrass}{rgb}{0.8, 0.58, 0.46}
\usepackage{algorithmic}
\usepackage[ruled,vlined,linesnumbered]{algorithm2e}
\usepackage{epsfig}
\usepackage{subfig}
\usepackage{mathrsfs,amsmath,amsthm,amssymb,mdwlist,bbm}

\usepackage{makecell}
\usepackage{multirow}

\longversion{
\usepackage[colorlinks=true,linkcolor=blue,citecolor=antiquebrass]{hyperref}
\usepackage{charter}
\usepackage{eulervm}
\usepackage{fullpage}
}

\allowdisplaybreaks

\renewcommand{\leq}{\leqslant}

\renewcommand{\ge}{\geqslant}
\renewcommand{\le}{\leqslant}

\shortversion{
	\setlength{\floatsep}{1pt}
	\setlength{\textfloatsep}{1pt}
	\setlength{\intextsep}{1pt}
	\setlength{\abovecaptionskip}{2pt}
	\setlength{\belowcaptionskip}{1pt}
	\setlength{\abovedisplayskip}{1pt}
	\setlength{\belowdisplayskip}{1pt}
}

\usepackage{enumerate}

\usepackage[shortlabels]{enumitem}

\newtheorem{lemma}{Lemma}
 \newtheorem{definition}{Definition}

 \setlist{nolistsep,leftmargin=*}

\usepackage{nicefrac}

\newcommand{\eps}{\ensuremath{\varepsilon}\xspace}
\renewcommand{\epsilon}{\eps}
\let\mydelta\delta
\renewcommand{\delta}{\ensuremath{\mydelta}\xspace}
\let\myalpha\alpha
\renewcommand{\alpha}{\ensuremath{\myalpha}\xspace}
\let\mystar\star
\renewcommand{\star}{\ensuremath{\mystar}}

\newcommand{\VC}{{\sc  Vertex Cover}\xspace}
\newcommand{\PVC}{{\sc  Partial Vertex Cover}\xspace}
\newcommand{\RSED}{{\sc Reducing Total Similarity}\xspace}
\newcommand{\ESED}{{\sc Eliminating Similarity}\xspace}
\newcommand{\RMSED}{{\sc Reducing Maximum Similarity}\xspace}
\newcommand{\USC}{{\sc Uniform Set Cover}\xspace}
\newcommand{\SC}{{\sc Set Cover}\xspace}

\newcommand{\pr}{\ensuremath{\prime}}

\newcommand{\el}{\ensuremath{\ell}\xspace}

\newcommand{\NP}{\ensuremath{\mathsf{NP}}\xspace}
\newcommand{\WOH}{\ensuremath{\mathsf{W[1]}}\text{-hard}\xspace}
\newcommand{\WTH}{\ensuremath{\mathsf{W[2]}}\text{-hard}\xspace}
\newcommand{\WO}{\ensuremath{\mathsf{W[1]}}\xspace}
\newcommand{\WT}{\ensuremath{\mathsf{W[2]}}\xspace}
\newcommand{\WP}{\ensuremath{\mathsf{W[P]}}\xspace}
\newcommand{\XP}{\ensuremath{\mathsf{XP}}\xspace}
\newcommand{\FPT}{\ensuremath{\mathsf{FPT}}\xspace}

\newcommand{\NPH}{\ensuremath{\mathsf{NP}}-hard\xspace}
\newcommand{\NPC}{\ensuremath{\mathsf{NP}}-complete\xspace}
\newcommand{\pNPH}{para-\ensuremath{\mathsf{NP}}-hard\xspace}

\newcommand{\similarity}{\text{{\sc{sim}}}\xspace}

\newcommand{\YES}{{\sc{yes}}\xspace}
\newcommand{\NO}{{\sc{no}}\xspace}

\newcommand{\true}{{\sc{true}}\xspace}
\newcommand{\false}{{\sc{false}}\xspace}

\newcommand{\RB}{\ensuremath{\mathbb R}\xspace}

\newcommand{\CC}{\ensuremath{\mathcal C}\xspace}
\newcommand{\DD}{\ensuremath{\mathcal D}\xspace}
\newcommand{\EE}{\ensuremath{\mathcal E}\xspace}
\newcommand{\FF}{\ensuremath{\mathcal F}\xspace}
\newcommand{\GG}{\ensuremath{\mathcal G}\xspace}

\newcommand{\MM}{\ensuremath{\mathcal M}\xspace}
\newcommand{\NN}{\ensuremath{\mathcal N}\xspace}
\newcommand{\OO}{\ensuremath{\mathcal O}\xspace}
\newcommand{\PP}{\ensuremath{\mathcal P}\xspace}

\renewcommand{\SS}{\ensuremath{\mathcal S}\xspace}
\newcommand{\TT}{\ensuremath{\mathcal T}\xspace}
\newcommand{\UU}{\ensuremath{\mathcal U}\xspace}
\newcommand{\VV}{\ensuremath{\mathcal V}\xspace}
\newcommand{\WW}{\ensuremath{\mathcal W}\xspace}
\newcommand{\XX}{\ensuremath{\mathcal X}\xspace}
\newcommand{\YY}{\ensuremath{\mathcal Y}\xspace}
\newcommand{\ZZ}{\ensuremath{\mathcal Z}\xspace}

\newtheorem{theorem}{\bf Theorem}
\newtheorem{proposition}{\bf Proposition}

\newtheorem{corollary}{\bf Corollary}

\usepackage{cleveref}

\crefname{theorem}{Theorem}{Theorems}
\crefname{observation}{Observation}{Observations}
\crefname{lemma}{Lemma}{Lemmas}
\crefname{corollary}{Corollary}{Corollaries}
\crefname{proposition}{Proposition}{Propositions}
\crefname{definition}{Definition}{Definitions}
\crefname{claim}{Claim}{Claims}
\crefname{table}{Table}{Tables}
\crefname{equation}{Inequality}{Inequalities}
\crefname{reductionrule}{Reduction rule}{Reduction rules}
\crefname{section}{Section}{Sections}

%% file: abstract.tex
\begin{abstract}
Node similarity measures quantify how similar a pair of nodes are in a network. These similarity measures turn out to be an important fundamental tool for many real world applications such as link prediction in networks, recommender systems etc. An important class of similarity measures are {\em local similarity measures}. Two nodes are considered similar under local similarity measures if they have large overlap between their neighboring set of nodes. Manipulating node similarity measures via removing edges is an important problem. This type of manipulation, for example, hinders effectiveness of link prediction in terrorists networks. All the popular computational problems formulated around manipulating similarity measures turn out to be \NPH. We, in this paper, provide fine grained complexity results of these problems through the lens of parameterized complexity. In particular, we show that some of these problems are fixed parameter tractable (\FPT) with respect to various natural parameters whereas other problems remain intractable (\WOH and \WTH in particular). Finally we show the effectiveness of our proposed FPT algorithms on real world datasets as well as synthetic networks generated using Barabasi-Albert and Erdos-Renyi models.
\end{abstract}

%% file: introduction.tex
\section{Introduction}

\begin{table*}[!htbp]
	\centering
	\renewcommand{\arraystretch}{1.2}
	\begin{tabular}{|c|c|c|}\hline
		Problem & Parameter& Results\\\hline\hline
		
		\multirow{4}{*}{\RSED} & $k$& \makecell{\WOH even for stars ~[\Cref{thm:rsed_k_woh}]}\\\cline{2-3}
		& $|\SS|$ & \FPT~[\Cref{fpt:rsed_S}]\\\cline{2-3}
		& $\Delta$ & $\OO(2^\Delta\text{poly}(n))$~[\Cref{fpt:esed_D}]\\\cline{2-3}
		& $\delta$ & \pNPH~[\Cref{thm:ph_avg_deg_esed}]\\\hline
		
		\multirow{4}{*}{\ESED} & $k$& \makecell{$\OO(1.2738^k+km)$~[\Cref{thm:etsed}],\\ $2$ approximation~[\Cref{cor:etsed_2aprox}],\\
			$2-\eps$ inapproximable under UGC~[\Cref{cor:etsed_2aprox_lb}]} \\\cline{2-3}
		& $|\SS|$ & \FPT~[\Cref{fpt:esed_S}]\\\cline{2-3}
		& $\Delta$ & $\OO(2^\Delta\text{poly}(n))$~[\Cref{fpt:rsed_D}]\\\cline{2-3}
		& $\delta$ & \pNPH~[\Cref{col:wh_avg_deg_rsed_rmsed}]\\\hline
		
		\multirow{4}{*}{\RMSED} & $k$& \makecell{\WTH~[\Cref{thm:rmsed_k_wth}]}\\\cline{2-3}
		& $|\SS|$ & \FPT~[\Cref{fpt:rmsed_S}]\\\cline{2-3}
		& $\Delta$ & \pNPH~[\Cref{thm:rmsed_D}]\\\cline{2-3}
		& $\delta$ & \pNPH~[\Cref{col:wh_avg_deg_rsed_rmsed}]\\\hline
	\end{tabular}
	\caption{Summary of results. We denote the maximum degree and average of input graph by $\Delta$ and $\delta$ respectively.}\label{tbl:summary}
\end{table*}

Analyzing social networks for uncovering hidden information has a wide range of applications in artificial intelligence (see~\cite{DBLP:conf/atal/SabaterS02,otte2002social,wang2007social,carrington2005models} and references therein for a variety of applications). One of the fundamental tools for social network analysis is the notion of a node similarity measure in networks. A node similarity measure is a function which quantifies the similarity between pairs of nodes of a given network. One such popular measures is the Jaccard similarity. The Jaccard similarity between two nodes $u$ and $v$ is the ratio of the number of common neighbors of $u$ and $v$ by the total number of nodes which are neighbor to at least one of $u$ or $v$. The Jaccard similarity measure belongs to a broad class of measures called {\em local} similarity measures. A similarity measure is called local if the similarity between two nodes depends only on their neighborhood (nodes which have an edge with at least one of the pair of nodes).

One of the most useful applications of network similarity measures is link prediction --- given a network, predict the edges which are likely to be added in the network in future~\cite{al2006link,liben2007link,wang2015link,zhou2009predicting}. The link prediction problem has a variety of applications from uncovering hidden links in a covert network~\cite{lim2019hidden} to recommendation systems~\cite{chen2005link,talasu2017link,campana2017recommender}.

Zhou et al. observed that the effectiveness of network similarity measures can be hampered substantially using various kind of attacks, edge deletion (that is hiding some of the existing edges) for example~\cite{Zhou:2019:ASL:3306127.3331707}. Indeed, they empirically showed that manipulation by deleting edges affects the effectiveness of important similarity measures in real world and synthetically generated networks. Fortunately, they showed that the computational task of manipulating any such measure is \NPC.

\subsubsection{Motivation:}

We abstract out the main computational challenge of manipulating any local similarity measure into three purely graph theoretic problems which does not depend directly on the particular similarity measures under consideration. Since the similarity between any two nodes decreases (under manipulation via edge deletion) under any local metric only when the number of their common neighbors decreases, the fundamental task in all our problems is to reduce the number of common neighbors of some given targeted pair of nodes. In our first problem, called \ESED, the input consists of a graph, a set of targeted pairs of nodes, a set \CC of edges which can be removed, and the maximum number $k$ of edges that we can delete. We need to compute if there exist $k$ edges in \CC whose removal ensures that every pair of given nodes has disjoint neighborhood. In some applications, covert networks for example, one may sustain some small amount of similarity even between targeted pairs of nodes. This motivates our second problem. In the \RSED problem, the input is the same as \ESED. We are additionally given an integer $t$ and we need to compute if there exist $k$ edges in \CC whose removal ensures that the sum of the number of common neighbors in the targeted pairs of nodes is at most $t$. One can immediately observe that it may happen that, although the sum of the number of common neighbors is reasonably small (compared with the number of pairs given), for few targeted pairs, the number of common neighbors (and thus similarity) could be high. We take care of this issue in the \RMSED problem. In this problem, the input is the same as the \RSED problem but the goal is to find (if exists) a set of at most $k$ edges in \CC whose removal ensures that the number of common neighbor between any targeted pair of nodes is at most $t$. The advantage of abstracting out the network similarity manipulation problem into graph theoretic problems is its universality. All our results, both algorithms and structural, apply uniformly for all local similarity measures. Obviously, the generality is achieved at a cost: it may be possible to have fine tuned specialized algorithms for specific similarity measures. 

\subsubsection{Contribution:}

We study parameterized complexity of our problems and also derive few results on approximation algorithms for some of our problems as a by product of our techniques. We consider three parameters: (i) the maximum number $k$ of edges that can be removed, (ii) The number $|\SS|$ of pairs of nodes whose similarity we wish to reduce or eliminate, (iii) the maximum degree of any vertex, and (iv) average degree of the graph. We exhibit FPT algorithms for all our problems parameterized by $|\SS|$. Whereas, for the parameter $k$, we show that only the \ESED problem admits an \FPT algorithm; the other two problems are intractable with respect to this parameter. We summarize our contribution in this paper in \Cref{tbl:summary}. We also show the following results.

\begin{itemize}
	\item Suppose, in the \ESED problem, there exists a set \WW of (important) vertices such that the set \SS of pairs of vertices between which we wish to eliminate similarity is the set of all pairs of vertices in \WW (that is, $\SS=\{\{u,v\}:u,v\in\WW,u\ne v\}$). Then we show that the \ESED problem is polynomial time solvable~[\Cref{thm:esed_spl_poly}].
	
	\item Experiments: We evaluate our proposed FPT algorithm on two synthetic and four real datasets from different genres. We show that our FPT algorithms produce significantly better results than some standard baselines. Moreover, our algorithms are efficient and significantly faster than the baselines.   

\end{itemize}

We observe that the \RMSED problem is the hardest followed respectively by \RSED and \ESED.

%% file: prelim.tex
\section{Preliminaries and Problem Formulation}

\subsection{Similarity Measures:}

Let $\GG=(\VV,\EE)$ be any undirected and unweighted graph. If not mentioned otherwise, we denote the number of vertices and the number of edges by $n$ and $m$ respectively. For any vertex $v\in\VV$, its neighborhood $\NN(v)$ is defined to be the set of all vertices which shares an edge with $v$, that is, we define $\NN(v)=\{u\in\VV: \{u,v\}\in\EE\}$. A similarity measure \similarity of a graph $\GG=(\VV,\EE)$ is a function $\similarity:\VV\times\VV\longrightarrow\RB_{\ge0}$. We usually do not talk about similarity of a vertex with itself and thus, for convenience, we define $\similarity(v,v)=0$ for every $v\in\VV$. We call a similarity function \similarity~{\em local} if, for every pair of vertices $u,v\in\VV$, (i) $\similarity(u,v)$ depends only on $\NN(u)$ and $\NN(v)$, (ii) $\similarity(u,v)=0$ if and only if $\NN(u)\cap\NN(v)=\emptyset$. The popular Jaccard similarity is local: The Jaccard similarity between $u$ and $v$ is 
$\similarity(u,v)= \frac{|\NN(u)\cap \NN(v)|}{|\NN(u)\cup \NN(v)|}$. The Adamic/Adar index is another important example of local similarity measure. Here, the similarity between two nodes $u$ and $v$ is defined as $\similarity(u,v)= \sum_{x\in \NN(u)\cap \NN(v)}\frac{1}{\log |\NN(x)|}$. If we restrict ourselves to local similarity measures, reducing similarity (by deleting edges) between two nodes under any local similarity measure boils down to reducing their common neighbors. With this observation, we define generic computational problems which captures the essence of the computational challenge of reducing similarity between two nodes under any local similarity measure.

\subsection{Problem Definition}

In first problem, we are given a set of pairs of nodes and we need to compute the set of minimum number of edges whose removal ensures that there is no common neighbor between any of the given pair of nodes.

\begin{definition}[\ESED]\label{probdef:esed}
	Given a graph $\GG=(\VV,\EE)$, a target set $\SS=\{\{x,y\}|x\in \VV,y \in \VV\}$ of pairs of vertices, a subset $\CC \subseteq\EE$ of candidate edges which can be deleted, and an integer $k$ denoting the maximum number of edges that can be deleted, compute if there exists a subset $\FF\subseteq \CC$ such that $|\FF|\leq k$ and no pair of vertices in \SS has any common neighbor in the graph $\GG\setminus\FF$. We denote an arbitrary instance of this problem by $(\GG,\SS,\CC,k)$.
\end{definition}

We generalize the \ESED problem where the goal is to
find a given budget number of edges such that the sum of the number of common neighbours in the given pairs of vertices is below some threshold. Formally, it is defined as follows.

\begin{definition}[\RSED]\label{probdef:rsed}
Given a graph $\GG=(\VV,\EE)$, a target set $\SS=\{\{x,y\}|x\in \VV,y \in \VV\}$ of pairs of vertices, a subset $\CC \subseteq\EE$ of candidate edges which can be removed, an integer $k$ denoting the maximum number of edges that can be deleted, and an integer $t$ denoting the target sum of number of common neighbors between vertices in \SS, compute if there exists a subset $\FF\subseteq \CC$ such that $|\FF|\leq k$ and sum of number of common neighbors between pairs of vertices in \SS is at most $t$ in the graph $\GG\setminus\FF$. We denote an arbitrary instance of this problem by $(\GG,\SS,\CC,k,t)$.
\end{definition}

In a solution to the \RSED problem, it may be possible that, although the sum of common neighbors between given pairs of vertices is below some threshold, there exist pairs of vertices in the given set which still share many neighbors. In the \RMSED problem defined below, the goal is to find $k$ edges to delete such that the maximum overlap between neighbours does not go beyond a threshold.

\begin{definition}[\RMSED]\label{probdef:rmsed}
Given a graph $\GG=(\VV,\EE)$, a target set $\SS=\{\{x,y\}|x\in \VV,y \in \VV\}$ of pairs of vertices, a subset $\CC \subseteq\EE$ of candidate edges which can be deleted, an integer $k$ denoting the maximum number of edges that can be deleted, and an integer $t$ denoting the target maximum of number of common neighbors between vertices in \SS, compute if there exists a subset $\FF\subseteq \CC$ such that $|\FF|\leq k$ and the number of common neighbors between every pair of vertices in \SS is at most $t$ in the graph $\GG\setminus\FF$. We denote an arbitrary instance of this problem by $(\GG,\SS,\CC,k,t)$.
\end{definition}

The following complexity theoretic relationship among our problems is straight forward.

\begin{proposition}\label{prop:connection}
	\ESED polynomial time many-to-one reduces to \RSED. \ESED polynomial time many-to-one reduces to \RMSED.
\end{proposition}

If not mentioned otherwise, we use $k$ to denote the number of edges that we are allowed to remove, $\SS$ to denote the target set.

\subsection{Parameterized complexity} 
A parameterized problem $\Pi$ is a 
subset of $\Gamma^{*}\times
\mathbb{N}$, where $\Gamma$ is a finite alphabet. A central notion is \emph{fixed parameter 
	tractability} (FPT) which means, for a 
given instance $(x,k)$, solvability in time $f(k) \cdot p(|x|)$, 
where $f$ is an arbitrary function of $k$ and 
$p$ is a polynomial in the input size $|x|$. There exists a hierarchy of complexity classes above FPT, and showing that a parameterized problem is hard for one of these classes is considered
evidence that the problem is unlikely to be fixed-parameter tractable. The main classes in this hierarchy are: $ \FPT  \subseteq \WO \subseteq \WT \subseteq \cdots \subseteq \WP \subseteq \XP.$ We now define the notion of parameterized reduction~\cite{CyganEtAl}.
\begin{definition}
	Let $A,B$ be parameterized problems.  We say that $A$ is {\bf \em fpt-reducible} to $B$ if there exist functions 
	$f,g:\mathbb{N}\rightarrow \mathbb{N}$, a constant $\alpha \in \mathbb{N}$ and 
	an algorithm $\Phi$ which transforms an instance $(x,k)$ of $A$ into an instance $(x',g(k))$ of $B$ 
	in time $f(k) |x|^{\alpha}$ 
	so that $(x,k) \in A$ if and only if $(x',g(k)) \in B$.
\end{definition}

To show W-hardness, it is enough to give a parameterized reduction from a known hard problem.

%% file: results.tex
\section{Algorithmic Results}

In this section, we present our polynomial time and FPT algorithms. Our first result shows that the \ESED problem is fixed parameter tractable parameterized by the number $k$ of edges that we are allowed to delete. In the interest of space, we omit few proofs. We mark them with $(\star)$.

\begin{theorem}\label{thm:etsed}
 There exists an algorithm for the \ESED problem which runs in time $\OO(1.2738^k+km)$.
\end{theorem}

\begin{proof}
 Let $(\GG=(\VV,\EE),\SS,\CC,k)$ be an arbitrary instance of \ESED. For every pair $\{x,y\}\in\SS$ and every $u\in\VV$ with $\{\{u,x\},\{u,y\}\}\in\EE$, if $|\{\{u,x\},\{u,y\}\}\cap\CC|=0$, then we output \NO; if $|\{\{u,x\},\{u,y\}\}\cap\CC|=1$ and $k\ge1$, then we remove the edge $\{\{u,x\},\{u,y\}\}\cap\CC$ and decrease $k$ by $1$. If no pair of vertices in \SS have common neighbor in the current graph, then we output \YES. Otherwise let us assume without loss of generality that, in the current instance $(\GG=(\VV,\EE),\SS,\CC,k)$, for every pair $\{x,y\}\in\SS$ and every $u\in\VV$ with $\{\{u,x\},\{u,y\}\}\in\EE$, we have $\{\{u,x\},\{u,y\}\}\subseteq\CC$.
 
 We now reduce the \ESED instance to a \VC $(\GG^\pr=(\VV^\pr,\EE^\pr),k^\pr)$. We have $\VV^\pr=\{v_e: e\in\EE\}$. For any two vertices $v_{e}, v_{f}\in\VV^\pr$, we have an edge between them in $\GG^\pr$, that is $\{v_e,v_f\}\in\EE^\pr$, if there exist vertices $a,b,c\in\VV$ such that $e=\{a,b\}, f=\{b,c\},$ and $\{a,c\}\in\SS$. We define $k^\pr=k$. We now claim that $(\GG,\SS,\CC,k)$ is a \YES instance of \ESED if and only if $(\GG^\pr,k^\pr)$ is a \YES instance of \VC.
 
 Let $(\GG,\SS,\CC,k)$ be a \YES instance of \ESED. Then there exists a subset $\FF\subseteq\EE$ of edges in \GG with $|\FF|\le k$ such that no pair of vertices in \SS has a common neighbor in $\GG\setminus\FF$. We claim that $\WW=\{v_e: e\in\FF\}$ is a vertex cover of $\GG^\pr$. Suppose not, then there exists at least one edge $\{v_e,v_f\}\in\EE^\pr$ which \WW does not cover. Then, by the construction of $\GG^\pr$, we have three vertices $a,b,c\in\VV$ such that $e=\{a,b\}, f=\{b,c\},$ and $\{a,c\}\in\SS$. This contradicts our assumption that no pair of vertices in \SS has a common neighbor in $\GG\setminus\FF$. Hence \WW forms a vertex cover of $\GG^\pr$. Moreover the size of \WW is $k$ which is same as $k^\pr$. Hence the \VC instance is a \YES instance.
 
 For the other direction, let $(\GG^\pr,k^\pr)$ is a \YES instance of \VC. Then there exists a vertex cover $\WW\subseteq\VV^\pr$ of $\GG^\pr$ of size at most $k^\pr$. Let us define $\FF=\{e\in\EE: v_e\in\WW\}$. We claim that there is no pair of vertices in \SS which has a common neighbor in $\GG\setminus\FF$. Suppose otherwise, then there exists a pair $\{a,c\}\in\SS$ which has a common neighbor $b\in\VV$ in $\GG\setminus\FF$. However, this implies that \WW does not cover the edge $\{v_{\{a,b\}}, v_{\{b,c\}}\}$ in $\GG^\pr$ which is a contradiction.
 
 We now describe our FPT algorithm. We first pre-process any instance $(\GG=(\VV,\EE),\SS,\CC,k)$ of \ESED and ensure that for every pair $\{x,y\}\in\SS$ and every $u\in\VV$ with $\{\{u,x\},\{u,y\}\}\in\EE$, we have $\{\{u,x\},\{u,y\}\}\subseteq\CC$. We then construct a corresponding instance $(\GG^\pr,k^\pr)$ of \VC as described above. We then compute a vertex cover \WW (if it exists) of $(\GG^\pr,k^\pr)$ and output the corresponding set of edges of \GG to be removed. If the vertex cover instance is a \NO instance, then we output \NO. We observe that the number of vertices in the vertex cover instance is the same as the number of edges of the \ESED instance. Now the claimed running time of our algorithm follows from the fact that there is an algorithm for the \VC problem which runs in time $\OO(1.2738^k+kn)$ where $n$ is the number of vertices in the \VC instance and $k$ is the size of the vertex cover we are seeking~\cite{DBLP:conf/mfcs/ChenKX06}.
\end{proof}

In the proof of \Cref{thm:etsed}, we exhibit aN approximation preserving reduction from \ESED to \VC. Since \VC admits a $2$ factor polynomial time approximation algorithm (see for example~\cite{DBLP:books/daglib/0004338}), we obtain the following result as an immediate corollary of \Cref{thm:etsed}.

\begin{corollary}\label{cor:etsed_2aprox}
	There exists a polynomial time algorithm for optimizing $k$ in the \ESED problem within a factor of $2$.
\end{corollary}

We next show that the \ESED problem is polynomial time solvable if there exists a subset $\WW\subseteq\VV$ of vertices such that the set \SS of given pairs is $\{\{u,v\}:u,v\in\WW,u\ne v\}$.

\begin{theorem}\label{thm:esed_spl_poly}\shortversion{$[\star]$}
	Suppose, in the \ESED problem, there exists a set \WW of (important) vertices such that the set \SS of pair of vertices between which we wish to eliminate similarity is the set of all pair of vertices in \WW (that is, $\SS=\{\{u,v\}:u,v\in\WW,u\ne v\}$). Then the \ESED problem is polynomial time solvable.
\end{theorem}

\longversion{
\begin{proof}
	Let $\GG=(\VV,\EE)$ be the input graph. We build the solution \XX (set of edges to be removed) iteratively. The set \XX is initialized to $\emptyset$. For every vertex $u\in\VV\setminus\WW$, we put all the edges between $u$ and every vertex in \WW except one edge in \XX. Let \MM be a maximum matching of $\GG[\WW]$. We put all the edges in $\EE[\WW]\setminus\MM$ in \XX where $\EE[\WW]$ is the set of edges in \GG with both end points in \WW. If the size of \XX exceeds the number $k$ of edges that we are allowed to delete, then we output \NO. Otherwise we output \YES with \XX being a set of at most $k$ edges whose deletion removes similarity between every pair of vertices in \WW. Clearly any solution would remove at least $|\XX|$ number of edges otherwise either there will be a vertex in $u\in\VV\setminus\WW$ which has edges to at least $2$ vertices in \WW (which is a contradiction) or there will be an induced path of length $2$ in $\GG[\WW]$ (which again is a contradiction since there exists pair of vertices in \WW having common neighbor). Hence, if the algorithm outputs \NO, the instance is indeed a \NO instance. On the other hand, if the algorithm outputs \YES, it discovers a set of at most $k$ edges whose removal ensures that no pair of vertices in \WW has any common neighbor. This concludes the correctness of the algorithm.
\end{proof}
}

We next focus on the parameter $|\SS|$. We use the following result by Lenstra to design our \FPT algorithms.

\begin{lemma}[Lenstra's Theorem~\cite{CyganEtAl}]\label{lenstra}
	There is an algorithm for computing a feasible as well as an optimal solution of an integer linear program which is fixed parameter tractable parameterized by the number of variables.
\end{lemma}

We first consider the \RSED problem.

\begin{theorem}\label{fpt:rsed_S}
	The \RSED problem parameterized by $|\SS|$ has a fixed parameter tractable algorithm.
\end{theorem}

\begin{proof}
	Let $(\GG=(\VV,\EE),\SS,\CC,k)$ be an arbitrary instance of \RSED. For every pair $\{x,y\}\in\SS$ and every $u\in\VV$ with $\{\{u,x\},\{u,y\}\}\in\EE$, if $|\{\{u,x\},\{u,y\}\}\cap\CC|=0$, then we output \NO; if $|\{\{u,x\},\{u,y\}\}\cap\CC|=1$ and $k\ge1$, then we remove the edge $\{\{u,x\},\{u,y\}\}\cap\CC$ and decrease $k$ by $1$. If no pair of vertices in \SS have common neighbor in the current graph, then we output \YES. Otherwise let us assume without loss of generality that, in the current instance $(\GG=(\VV,\EE),\SS,\CC,k)$, for every pair $\{x,y\}\in\SS$ and every $u\in\VV$ with $\{\{u,x\},\{u,y\}\}\in\EE$, we have $\{\{u,x\},\{u,y\}\}\subseteq\CC$.
	
	For two vertices $u$ and $v$ of \GG, we say that $u$ and $v$ are of ``same type" if we have the following --- for every pair $\{x,y\}\in\SS$ of vertices in \SS, $u$ is a common neighbor of $x$ and $y$ if and only if $v$ is a common neighbor of $x$ and $y$. Since $|\SS|=\el$, we observe that there can be at most $2^\el$ different types of vertices in \GG since types are in one-to-one correspondence with the power set $2^\SS$ of \SS. Let $\TT=\{\mu(\XX): \XX\subseteq\SS\}$ be the set of all types in \GG. For each type $\mu(\XX)\in\TT$, let $n(\XX)$ be the number of vertices of type $\mu(\XX)$ in \GG --- we observe that $n(\XX)$ can be computed in polynomial time from the given graph \GG for every type $\mu(\XX)\in\TT$. Let $v$ be any vertex in \GG of type $\mu(\XX)\in\TT$. We observe that the vertex $v$ can ``participate" in any optimal solution in only $4^{|\XX|}$ possible ways: for each pair $\{x,y\}\in\XX$ such that the vertex $v$ is a common neighbor of both $x$ and $y$, exactly one of the following $4$ events happen -- (i) both the edges $\{v,x\}$ and $\{v,y\}$ belong to the optimal set of edges (call it optimal solution), (ii) only $\{v,x\}$ belongs to the optimal solution, (iii) only $\{v,y\}$ belongs to the optimal solution, and (iv) none of $\{v,x\}$ and $\{v,y\}$ belong to the optimal solution. So vertex of each type $\mu(\XX)\in\TT$ can ``participate" in the optimal solution in at most $4^{|\XX|}$ ways; we abstractly define the set of all possible ways of participation of vertices of type $\mu(\XX)$ by $\PP(\XX)$. We will now formulate the \RSED problem using an integer linear program (ILP). For each type $\mu(\XX)\in\TT$ and each participation type $P\in\PP(\XX)$, let the variable $\ZZ(\XX;P)$ denote the number of vertices of type $\mu(\XX)$ which participate in the optimal solution like $P$. We use the variable $\YY(\{x,y\})$ in ILP to denote the number of common neighbors of $x$ and $y$ for $\{x,y\}\in\SS$ after removing the optimal set of edges. For each type $\mu(\XX)$ and every participation $P\in\PP(\XX)$, let $\lambda(\XX,P)$ denote the number of edges which gets removed in $P$; $\lambda(\XX,P)$ is a polynomial time computable fixed (depends on the input graph \GG only) quantity. We write the following ILP.
	\begin{align}
	 &\sum_{\{x,y\}\in\SS} \YY(\{x,y\}) \le t\label{eq:ilp1}\\
	 &\sum_{\mu(\XX)\in\TT,P\in\PP(\XX)} \lambda(\XX,P) \ZZ(\XX;P) \le k\label{eq:ilp2}\\
	 &\sum_{P\in\PP(\XX)}\ZZ(\XX;P) = n(\XX), & \forall \XX\in\TT\label{eq:ilp3}\\
	 &\YY(\{x,y\}) = \sum_{\mu(\XX)\in\TT,\{x,y\}\in\XX}n(\XX) - \nonumber\\
	 &\sum_{\substack{P\in\PP(\XX), \text{ at least}\\ \text{one edge on }x \text{ or }y\\ \text{gets removed in }P}} \ZZ(\XX;P) & \forall \{x,y\}\in\SS\label{eq:ilp4}
	\end{align}
	
	\Cref{eq:ilp1} along with \Cref{eq:ilp4} ensure that the sum of the number of common neighbors between pairs of vertices in \SS is at most $t$. \Cref{eq:ilp2} ensures that the number of edges deleted is at most $k$. From the discussion above, it follows that the \RSED instance is a \YES instance if and only if the above ILP is feasible.	Since the number of variables is $\OO(\el 2^\el 4^\el)=\OO(\el 8^\el)$, the result follows from \Cref{lenstra}.
\end{proof}

Due to \Cref{prop:connection}, \Cref{fpt:rsed_S} immediately gives us the following corollary.

\begin{theorem}\label{fpt:esed_S}
	The \ESED problem parameterized by $|\SS|$ has a fixed parameter tractable algorithm.
\end{theorem}

Also the idea of \Cref{fpt:rsed_S} can be analogously used to obtain a fixed parameter tractable algorithm for the \RMSED problem.

\begin{theorem}\label{fpt:rmsed_S}\shortversion{$[\star]$}
	The \RMSED problem parameterized by $|\SS|$ has a fixed parameter tractable algorithm.
\end{theorem}

\longversion{
\begin{proof}
	We replace \Cref{eq:ilp1} with the following two inequalities. The rest of the argument is analogous to the proof of \Cref{fpt:rsed_S}.
	\begin{align*}
		& \Gamma \le t\\
		& \YY(\{x,y\}) \le \Gamma & \forall \{x,y\}\in\SS
	\end{align*}
\end{proof}
}

We next consider maximum degree of any node as our parameter. We show that \RMSED admits an \FPT algorithm with respect to maximum degree as parameter.

\begin{theorem}\label{fpt:rsed_D}
 There is an algorithm for the \RSED problem with running time $\OO(2^\Delta\text{poly}(n))$ where $\Delta$ is the maximum degree of the input graph.
\end{theorem}

\begin{proof}
 Let $(\GG=(\VV,\EE),\SS,\CC,k,t)$ be an arbitrary instance of \RMSED. For each pair $\{x,y\}\in\SS$, we observe that $x$ and $y$ can have at most $\Delta$ common neighbors. Let $\NN_\EE(x)$ and $\NN_\EE(y)$ be the set of edges incident on respectively $x$ and $y$. Then we have $|\NN_\EE(x)\cup\NN_\EE(y)|\le 2\Delta$. Suppose $\VV=\{v_j: j\in[n]\}$ and for $i\in[n]$, we define $\VV_i=\{v_j: j\in[i]\}$ and $\GG_i=\GG[\VV_i]$. We now describe a dynamic programming based algorithm. Our dynamic programming table \TT is a Boolean table indexed by the set $\{(i,k^\pr,t^\pr): i\in[n], k^\pr\in[k], t^\pr=[t]\}$. We define $\TT(i,k^\pr,t^\pr)$ to be \true if and only if there exists a set $\FF_i\subseteq\EE[\GG_i]$ of $k^\pr$ edges whose removal from \GG makes the total number of common neighbor in $\VV_i$ between pairs of vertices in \SS to be at most $t^\pr$. Formally, $\TT(i,k^\pr,t^\pr)=\true$ if and only if $\exists \FF_i\subseteq\EE[\GG_i], |\FF_i|\le k^\pr$ such that $\sum_{\{x,y\}\in\SS} \NN(x)\cap\NN(y)\cap\VV_i^\pr \le t^\pr$ in $\GG\setminus\FF_i$. We initialize the table entries $\TT(1,k^\pr,t^\pr)$ to be \false for every $k^\pr\in[k]$ and $t^\pr\in[t]$ and initialize $\TT(1,k^\pr,0)$ to be \true for every $k^\pr\in[k]$. To update an entry $\TT(i,k^\pr,t^\pr)$, we guess the edges incident on $v_i$ that will be part of an optimal solution. Formally, we update $\TT(i,k^\pr,t^\pr)$ as follows. For any $X\subseteq N(v_i)$, we define $\Gamma(X,v_i)$ to be $|\{\{x,y\}\in\SS: x\in X \text{ or } y\in X\}|$.
 \begin{align*}
 \TT(i,k^\pr,t^\pr)
 = \bigvee_{\substack{X\subseteq N(v_i)}}
 \TT\left(i-1, k^\pr-|X|, t^\pr - \Gamma(X,v_i)\right)
 \end{align*}
 For convenience, we define the logical OR of no variables to be \false. We output that the \RMSED instance is a \YES instance if and only if $\TT(n,k,t)$ is \true. The correctness of our algorithm is immediate from our dynamic programming formulation and update rule. We observe that our dynamic programming table has $nkt$ entries each of which can be updated in $\OO(2^\Delta\text{poly}(n))$ time. Hence the running time of our algorithm is $\OO(2^\Delta\text{poly}(n))$.
\end{proof}

Due to \Cref{prop:connection}, \Cref{fpt:rsed_D} immediately implies existence of an \FPT algorithm for the \ESED problem parameterized by the maximum degree $\Delta$ of the graph.

\begin{theorem}\label{fpt:esed_D}
	There is an algorithm for the \ESED problem with running time $\OO(2^\Delta\text{poly}(n))$ where $\Delta$ is the maximum degree of the input graph.
\end{theorem}

\section{Hardness Results}

In this section we present our algorithmic hardness results. We begin with showing that the \RSED problem is \WOH parameterized by the number $k$ of edges that we are allowed to delete even for star graphs. For that, we present an fpt-reduction from the \PVC problem parameterized by budget to the \RSED problem. The \PVC problem is defined as follows.

\begin{definition}[\PVC]
 Given a graph \GG and two integers $k$ and $s$, compute if there exist $k$ vertices which cover at least $s$ edges. We denote an arbitrary instance of \PVC by $(\GG,k,s)$.
\end{definition}

We know that the \PVC problem parameterized by the number $k$ of vertices that we are allowed to choose is \WOH~\cite{CyganEtAl}.

\begin{theorem}\label{thm:rsed_k_woh}
	The \RSED problem parameterized by $k$ is \WOH even for stars.
\end{theorem}

\begin{proof}
	We exhibit an fpt-reduction from \PVC parameterized by the number of vertices that we are allowed to the \RSED problem parameterized by $k$. Let $(\GG=(\VV,\EE),k,s)$ be an arbitrary instance of \PVC. We construct an instance $(\GG^\pr=(\VV^\pr,\EE^\pr),\SS,\CC,k^\pr,t)$ of \RSED as follows. Let $m$ be the number of edges in \GG.
	\begin{align*}
		\VV^\pr &= \{u_v: v\in\VV\} \cup \{r\}\\
		\EE^\pr &= \{\{r,u_x\},\{r,u_y\}:\{x,y\}\in\EE \}\\
		\SS &= \{\{u_x,u_y\}: \{x,y\}\in\EE\}\\
		\CC &= \EE^\pr, k^\pr = k, t=m-s
	\end{align*}
	
	We observe that the number of pairs in \SS is $m$. Also, every pair of vertices in \SS has exactly one common neighbor namely $r$. We claim that the \PVC instance is a \YES instance if and only if the \RSED instance is a \YES instance.
	
	In one direction, let us assume that the \PVC instance is a \YES instance. Let $\WW\subseteq\VV$ be a subset of vertices which covers at least $s$ edges of \GG. Then the set $\FF=\{\{r,u_v\}: v\in\WW\}\subseteq\EE^\pr$ of edges makes the common neighborhood of every pair in $\{\{u_x,u_y\}:\{x,y\}\in\EE, x\in\WW \text{ or } y\in\WW\}$ empty. Since \WW covers at least $s$ edges in \GG, it follows that the sum of number of common neighbors between vertices in \SS in $\GG\setminus\FF$ is at most $m-s$ and thus the \RSED instance is a \YES instance.
	
	On the other direction, let us assume that the \RSED instance is a \YES instance. Let $\FF\subseteq\EE^\pr$ be a set of edges such that the sum of the number of common neighbors between vertices in \SS in $\GG\setminus\FF$ is at most $m-s$. Let us consider $\WW=\{v\in\VV:\{r,u_v\}\in\FF\}\subseteq\VV$. It follows that \FF covers every edge in $\{\{x,y\}\in\EE:\{r,u_x\}\in\FF\text{ or }\{r,u_y\}\in\FF\}$ which has at least $s$ edges since the sum of number of common neighbors between vertices in \SS in $\GG\setminus\FF$ is at most $m-s$. Hence the \PVC instance is a \YES instance.\longversion{
	
	Since the above reduction is an fpt-reduction, the result follows.}
\end{proof}

In the proof of \Cref{thm:rsed_k_woh}, we also exhibit an approximation preserving reduction from \VC (set $s=m$) to the optimization version of the \ESED problem where the goal is to remove a minimum number of edges. Since \VC is known to be inapproximable within factor $(2-\eps)$ for any $\eps>0$ in polynomial time under Unique Games Conjecture (UGC), we immediately have the following corollary (see \cite{DBLP:books/daglib/0004338} for example).

\begin{corollary}\label{cor:etsed_2aprox_lb}
	For every $\eps>0$, there does not exist any polynomial time algorithm for approximating $k$ for the \ESED problem within a factor of $(2-\eps)$ unless UGC fails.
\end{corollary}

We next show that the \RMSED problem is \WTH parameterized by $k$. Towards that, we use the a specialization of the set cover problem where every element of the universe appears in the same number of sets. We first show that the set cover problem with this assumption still continues to be \WTH and then present an \FPT-reduction from it to our problem.

\longversion{
\begin{definition}[\USC]
	Given an universe \UU, a collection $\DD\subseteq 2^\UU$ of subsets of \UU such that every element $u\in\UU$ appears in the same number of sets in \SS, and a budget $b$, compute if there exists a sub-collection $\WW\subseteq\DD$ such that (i) $|\WW|\le b$ and (ii) $\cup_{A\in\WW} A = \UU$. We denote an arbitrary instance of \USC by $(\UU,\DD,b)$.
\end{definition}

The \SC problem is the same as the \USC problem except the fact that elements in the universe \UU can belong to any number of sets in \SS in \SC. It is known that the \SC problem parameterized by $b$ is \WTH~\cite{CyganEtAl}.

\begin{proposition}
	The \USC problem parameterized by $b$ is \WTH.
\end{proposition}

\begin{proof}
	We exhibit an fpt-reduction from \SC to \USC. Let $(\UU,\DD,b)$ be an instance of \SC. Without loss of generality, we assume that, for every element $u\in\UU$, there exists a set $X\in\DD$ such that $u\in X$. We construct the following instance $(\UU^\pr,\DD^\pr,b^\pr)$ of \USC. For any element $u\in\UU$, let $f_u$ be the number of sets in \DD where $u$ belongs and $|\DD|=\el$.
	\begin{align*}
		\UU^\pr &= \UU\\
		\DD^\pr &= \DD \cup_{u\in\UU} \uplus_{i=f_u}^\el \{\{u\}\}\\
		b^\pr &= b 
	\end{align*}
	
	The equivalence of the two instances are straight forward and we defer its proof to the full version of the paper.
\end{proof}

We now prove that \RMSED is \WTH parameterized by the number $k$ of vertices that we are allowed to delete. For that, we exhibit an fpt-reduction from the \USC problem parameterized by the budget to the \RMSED problem parameterized by $k$.
}

\begin{theorem}\label{thm:rmsed_k_wth}\shortversion{$[\star]$}
	The \RMSED problem parameterized by $k$ is \WTH even for bipartite graphs of radius $2$.
\end{theorem}

\longversion{
\begin{proof}
	We exhibit an fpt-reduction from \USC to \RMSED. Let $(\UU,\SS,b)$ be an arbitrary instance of \USC. Let $f$ be the number of sets that every element in \UU belongs. We consider the following instance $(\GG=(\VV,\EE),\SS,\CC,k,t)$ of \RMSED.
	\begin{align*}
		\VV &= \{r\} \cup \{x_u: u\in\UU\}\cup \{y_D: D\in\DD\}\\
		\EE &= \{\{r,y_D\}:D\in\DD\}\cup\{\{x_u,y_D\}:u\in D\}\\
		\SS &= \{\{r,x_u\}:u\in\UU\}\\
		\CC &= \EE^\pr, k=b, t=f-1
	\end{align*}
	We claim that the \USC instance is a \YES instance if and only if the \RMSED instance is a \YES instance.
	
	In one direction, let us assume that the \USC instance is a \YES instance. Let $\WW\subseteq\DD$ forms a set cover for \UU and $|\WW|\le b$. We claim that $\FF=\{\{r,y_D\}:D\in\WW\}\subseteq\EE^\pr$ is a solution for the \RMSED instance. To see this, we consider any pair $\{r,x_u\}\in\SS$. By the definition of $f$, there are $f$ common neighbors of $r$ and $x_u$ in \GG. Since \WW forms a set cover, the number of common neighbors of $r$ and $x_u$ in $\GG\setminus\FF$ is at most $f-1$. We also have $|\FF|\le b=k$. This proves that the \RMSED instance is a \YES instance.
	
	On the other direction, let us assume that the \RMSED instance is a \YES instance. Let $\FF\subseteq\EE^\pr$ be a set of edges such that (i) $|\FF|\le k$ and (ii) the number of common neighbors in every pair of vertices in \SS is at most $t=f-1$. We consider the sub-collection $\WW=\{ D\in\DD: \{r,y_D\}\in\FF\text{ or }\{x_u,y_D\}\in\FF\text{ for some } u\in\UU \}\subseteq\DD$. Since $|\FF|\le k$, we have $|\WW|\le |\FF|\le k=b$. We claim that \WW forms a set cover for \UU. Suppose not, then there exists an element $u\in\UU$ that is not covered by \WW. Then the number of common neighbors of $r$ and $x_u$ in $\GG\setminus\FF$ is $f$ which is a contradiction. Hence the \USC instance is a \YES instance.
	
	Since the above reduction is an fpt-reduction, the result follows.
\end{proof}
}

Our last result of this section is that the \RMSED problem is \pNPH parameterized by the maximum degree $\Delta$ of the graph. For that, we use the well known result that the \VC problem is \NPC even for $3$ regular graphs~\cite{garey1974some}.

\begin{theorem}\label{thm:rmsed_D}\shortversion{$[\star]$}
	The \RMSED problem is \NPC even if the degree of every vertex in the input graph is at most $7$.
\end{theorem}

\longversion{
\begin{proof}
	The \RMSED problem is clearly in \NP. To prove \NP-hardness, we reduce from an arbitrary instance $(\GG=(\VV,\EE),k)$ of \VC on $3$-regular graph. We construct the following instance $(\GG^\pr=(\VV^\pr,\EE^\pr), \SS, \CC, k^\pr,t)$ of \RMSED.
	\begin{align*}
		\VV^\pr &= \{x_u, y_u: u\in\UU\}\\
		\EE^\pr &= \{\{x_u,x_v\}, \{x_u, y_v\}, \{x_v, y_u\}: \{u,v\}\in\EE\}\\
		&\cup \{\{x_u,y_u\}: u\in\VV\}\\
		\SS &= \{\{y_u, y_v\}: \{u,v\}\in\EE\}\\
		\CC &= \EE^\pr\\
		k^\pr &= k, t=1
	\end{align*}
	We now claim that the two instances are equivalent. Since the degree of every vertex in \GG is $3$, it follows that the degree of any vertex in $\GG^\pr$ is at most $7$. In one direction let us assume that the \VC instance is a \YES instance. Let $\WW\subseteq\VV$ be a vertex cover of cardinality $k$. We claim that, after removing every edge in the set $\FF=\{\{x_w,y_w\}:w\in\WW\}$, the number of common neighbors between every pair of vertices in \SS is at most $1$. Suppose not, then there exists a pair $\{y_u, y_v\}$ which has at least $2$ neighbors in $\GG^\pr\setminus\FF$. Then, it follows that both $\{x_u,y_u\}\notin\FF$ and $\{x_v,y_v\}\notin\FF$. However this implies that \WW does not cover the edge $\{u,v\}$ in \GG which contradicts our assumption that \WW is a vertex cover for \GG. Hence the \RMSED instance is a \YES instance.
	
	For the other direction, let us assume that there exists a subset $\FF\subseteq\EE^\pr$ of edges in $\GG^\pr$ such that, in $\GG^\pr\setminus\FF$, the number of common neighbors between every pair of vertices in \SS is at most $1$. Let us consider a subset $\WW=\{w\in\VV: \text{ an edge incident on } x_w \text{ belongs to } \FF\}$. Since $|\FF|\le k^\pr=k$, we have $|\WW|\le k$. We claim that \WW forms a vertex cover for \GG. Suppose not, then there exists an edge $\{u,v\}\in\EE$ which the set \WW does not cover. Then it follows that the both $x_u$ and $x_v$ are common neighbor of the pair $\{y_u, y_v\}$ of vertices. However, this contradicts our assumption about \FF. Hence \WW forms a vertex cover for \GG and the \VC instance is a \YES instance. 
\end{proof}
}

We finally consider the average degree $\delta$ of the graph as our parameter. We show that \ESED is \pNPH parameterized by $\delta$.

\begin{theorem}\label{thm:ph_avg_deg_esed}
	The \ESED problem is \pNPH parameterized by average degree $\delta$ of the input graph.
\end{theorem}

\longversion{
\begin{proof}
	We reduce any instance of \ESED to another instance of \ESED where, in the reduced instance, the average degree is constant. Let $\Pi=(\GG=(\VV,\EE),\SS,\CC,k)$ be an arbitrary instance of \ESED. We consider the following instance $\Pi^\pr=(\GG^\pr=(\VV^\pr,\EE^\pr),\SS^\pr,\CC^\pr,k^\pr)$ of \ESED. Let $|\VV|=n$ and $v$ be an arbitrary vertex of \GG.
	\begin{align*}
		\VV^\pr &= \VV \cup \{w_i: i\in[n^2]\}\\
		\EE^\pr &= \EE \cup \{\{w_i,w_{i+1}\}: i\in[n^2-1]\} \cup \{\{v,w_1\}\}\\
		\SS^\pr &= \SS, \CC^\pr = \CC, k^\pr = k
	\end{align*}
	
	We observe that the average degree of $\GG^\pr$ is at most $\frac{2n^2}{n^2+n}\le2$. We now claim that the two instances are equivalent. In one direction, if $\FF\subseteq\EE$ is a solution for $\Pi$ then \FF is a solution for $\Pi^\pr$ too. On the other hand, if $\FF^\pr\subseteq\EE^\pr$ is a solution for $\Pi^\pr$, then $\FF^\pr\cap\EE$ is a solution for $\Pi$ too.
\end{proof}
}

Due to \Cref{prop:connection}, we immediately have the following from \Cref{thm:ph_avg_deg_esed}.

\begin{corollary}\label{col:wh_avg_deg_rsed_rmsed}
	Both the \RSED and \RMSED problems are \pNPH parameterized by average degree $\delta$ of the input graph.
\end{corollary}

%% file: exp.tex
\begin{figure*}[t]
    \centering
    \subfloat[Web ]{\includegraphics[width=0.4\textwidth]{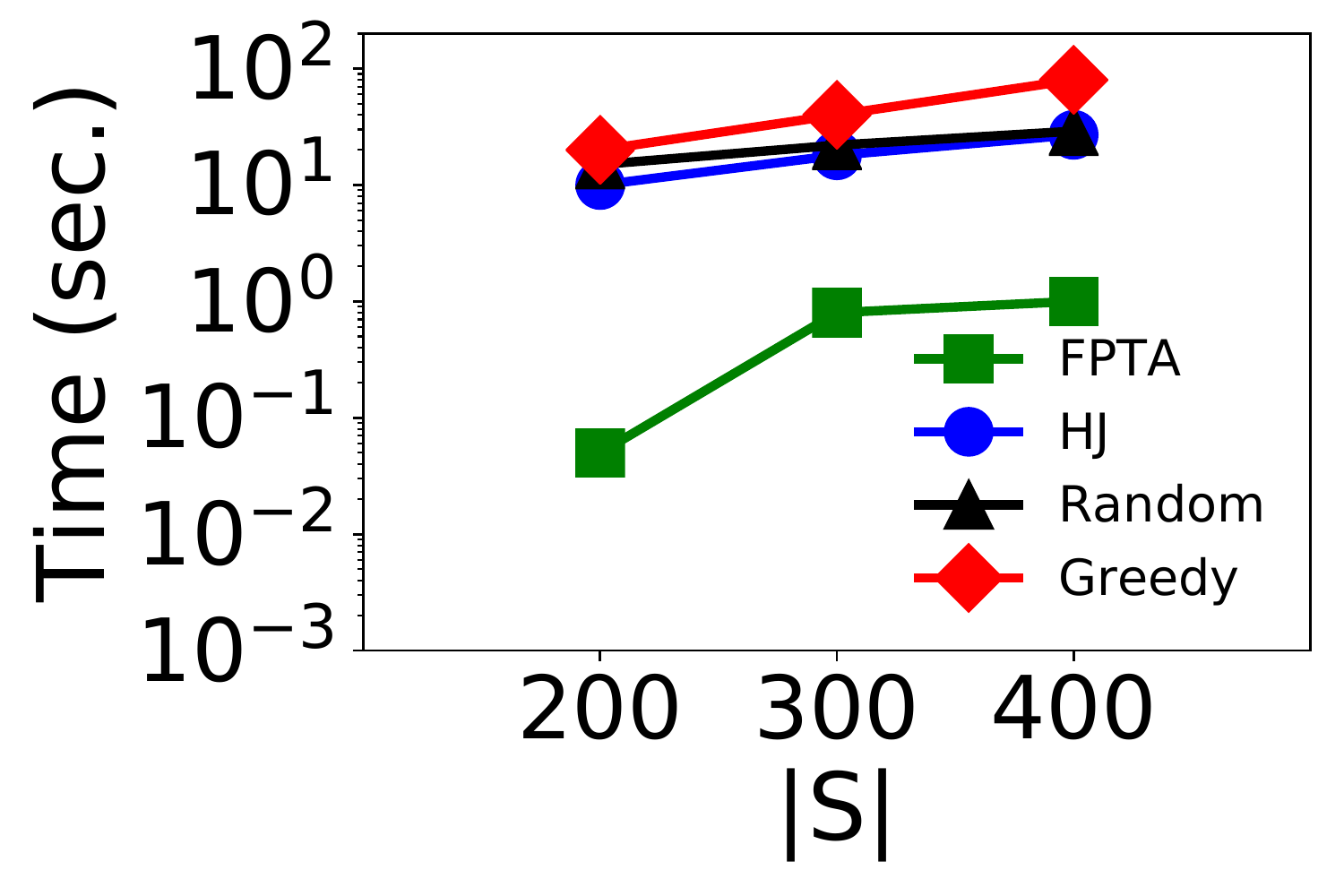}\label{fig:web_time}}
     \subfloat[Power]{\includegraphics[width=0.4\textwidth]{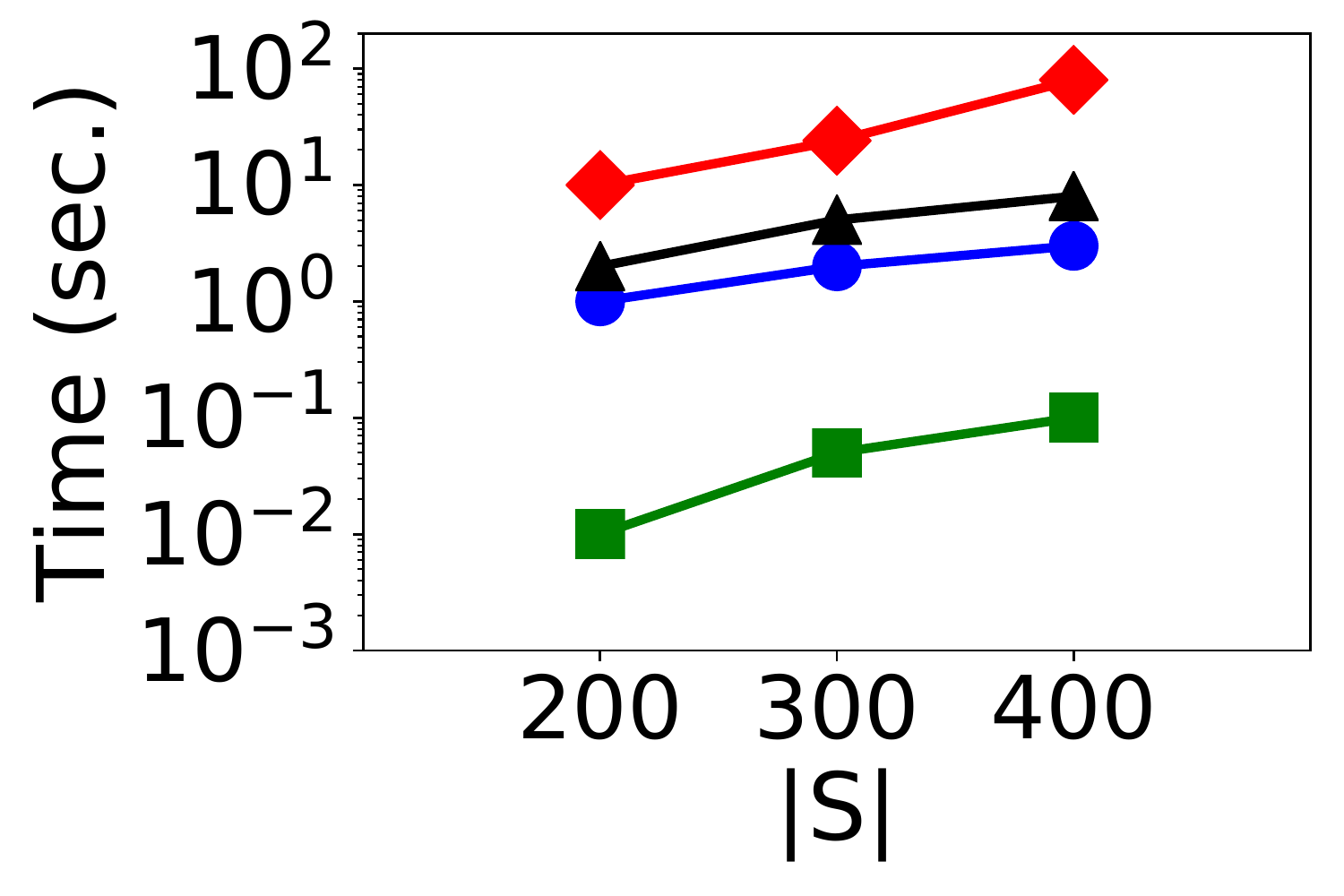}\label{fig:power_time}} \\
    \subfloat[Social]{\includegraphics[width=0.4\textwidth]{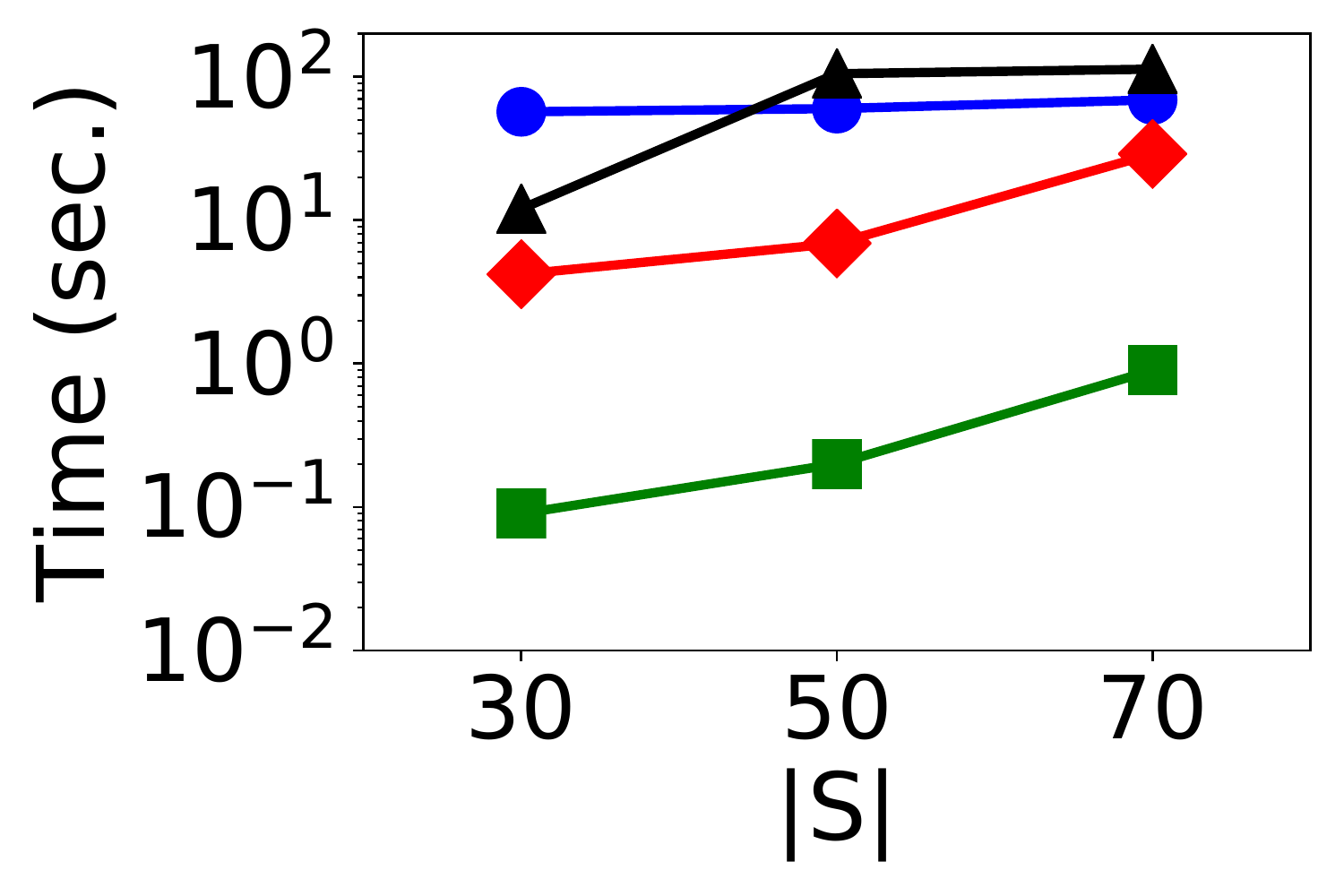}\label{fig:social_time}} 
    \subfloat[Road]{\includegraphics[width=0.4\textwidth]{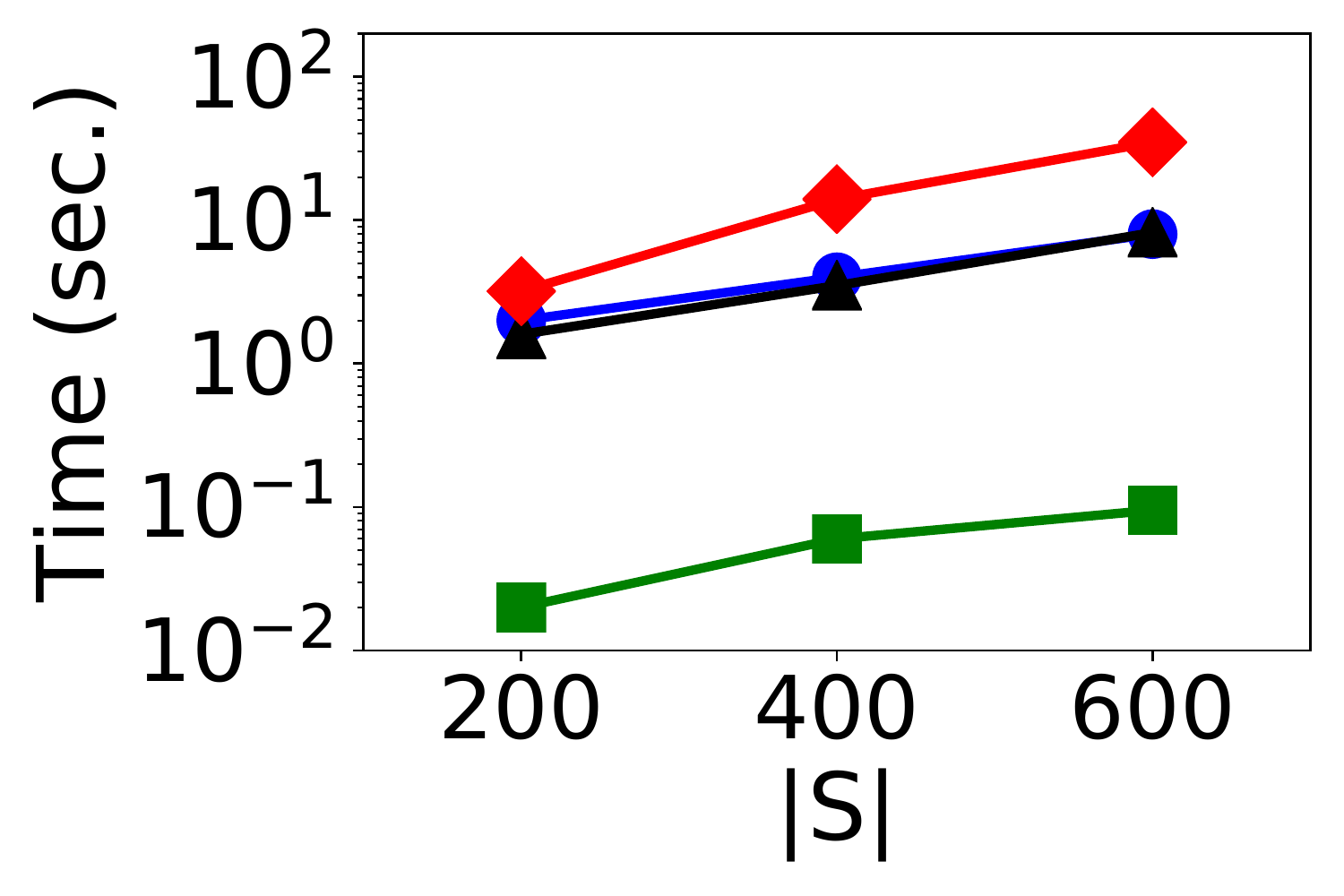}\label{fig:road_time}} 
    \caption{ Time taken by different algorithms: Our algorithm FPTA outperforms other baselines in terms of efficiency. \label{fig:baselines_time}}
\end{figure*}

\section{Experimental Results}\label{sec:exp}

In this section, we evaluate the performance of our FPT algorithm in \Cref{thm:etsed} using real and synthetic networks for the optimization version of the \ESED problem. We call our solution as FPTA. We implement our solutions in python and executed on a $3.30$GHz Intel core with $30$ GB RAM.

\begin{table}[t]
	\centering
	
	\begin{tabular}{|c|c|c|c|}
	\hline
		\textbf{Dataset}& Type &$|V|$ & $|E|$\\\hline
		power & Power & 662 & 1.5k \\\hline
		hamsterster & Social & 2.4k & 16.6k \\\hline
		euroroad & Road & 1.1k & 1.4k \\\hline
		web-edu & Web & 3k & 6.4k \\\hline
	\end{tabular}
	\caption{Description of Datasets}\label{tbl:dataset}
\end{table}

\textbf{Datasets: } We use synthetic graphs from two well-studied models: (a) Barabasi-Albert (BA)~\cite{barabasi1999emergence} and (b) Erdos-Renyi (ER)~\cite{erdHos1960evolution}. While the BA model has ``small-world'' property and scale free degree distribution, ER does not have any of these properties. We generate both the datasets of one thousand vertices (dataset of similar size as in \cite{Zhou:2019:ASL:3306127.3331707}) and approximately two thousands edges. The real datasets are from different genres: Web, social, road and power networks. Table \ref{tbl:dataset} shows the statistics. The datasets are available online\footnote{http://networkrepository.com}.

\textbf{Baselines: }We compare our algorithm (FPTA) with two baselines. Note that our algorithm produces optimal results. \textbf{(1) Greedy: }Our first baseline is the greedy baseline that is used in many graph combinatorial problems \cite{medya2018noticeable,silva2015hierarchical}. It selects an edge in each iteration which decreases a maximum number of common neighbors between the given pairs of nodes and removes it from the graph. \textbf{(2) High Jaccard Similarity (HJ):} It selects top edges based on the similarity of endpoints of the edge to delete until every given pair vertices have disjoint neighborhood. \textbf{(3) Random: } It iteratively selects random edges to delete until the total similarity of the target pairs becomes zero. The performance metric of these algorithms is the number of edges being deleted to remove similarity for all the given pairs. Hence, the quality is better when the number of edges is lower.

 In the experiments, we choose the target pairs ($S$) randomly from all the pairs of vertices. The size of $S$ is varied depending of the size and nature of the datasets.

\Cref{thm:etsed} shows that an FPT algorithm for \ESED parameterized by the number of edges that we are allowed to delete. Thus our algorithm always outputs a minimum set of edges. However, we evaluate the efficiency of our method in terms of both quality and running time. Table \ref{table:quality_random} shows the results varying $|S|$ (the number of target pairs) on four different real datasets. The results in synthetic graphs also have similar trend. The optimal set of edges are quite low compared to the Random and HJ baselines. We also find that the greedy algorithm also produces nearly optimal results. However it is quite time consuming.
 We show the results regarding the time taken by different algorithms in Figure \ref{fig:baselines_time} (real graphs) and Figure \ref{fig:baselines_time_synthetic} (synthetic graphs). The Y-axis is in logarithmic scale. In all six datasets, the time taken by our algorithm FPTA is at least two order faster. The most time consuming algorithm is the greedy algorithm. While Random and HJ are faster than Greedy in most of the cases, it produces much worse results compared to FPTA (Table \ref{table:quality_random}). 
 

\begin{figure}[t]
    \centering
    \subfloat[ER ]{\includegraphics[width=0.4\textwidth]{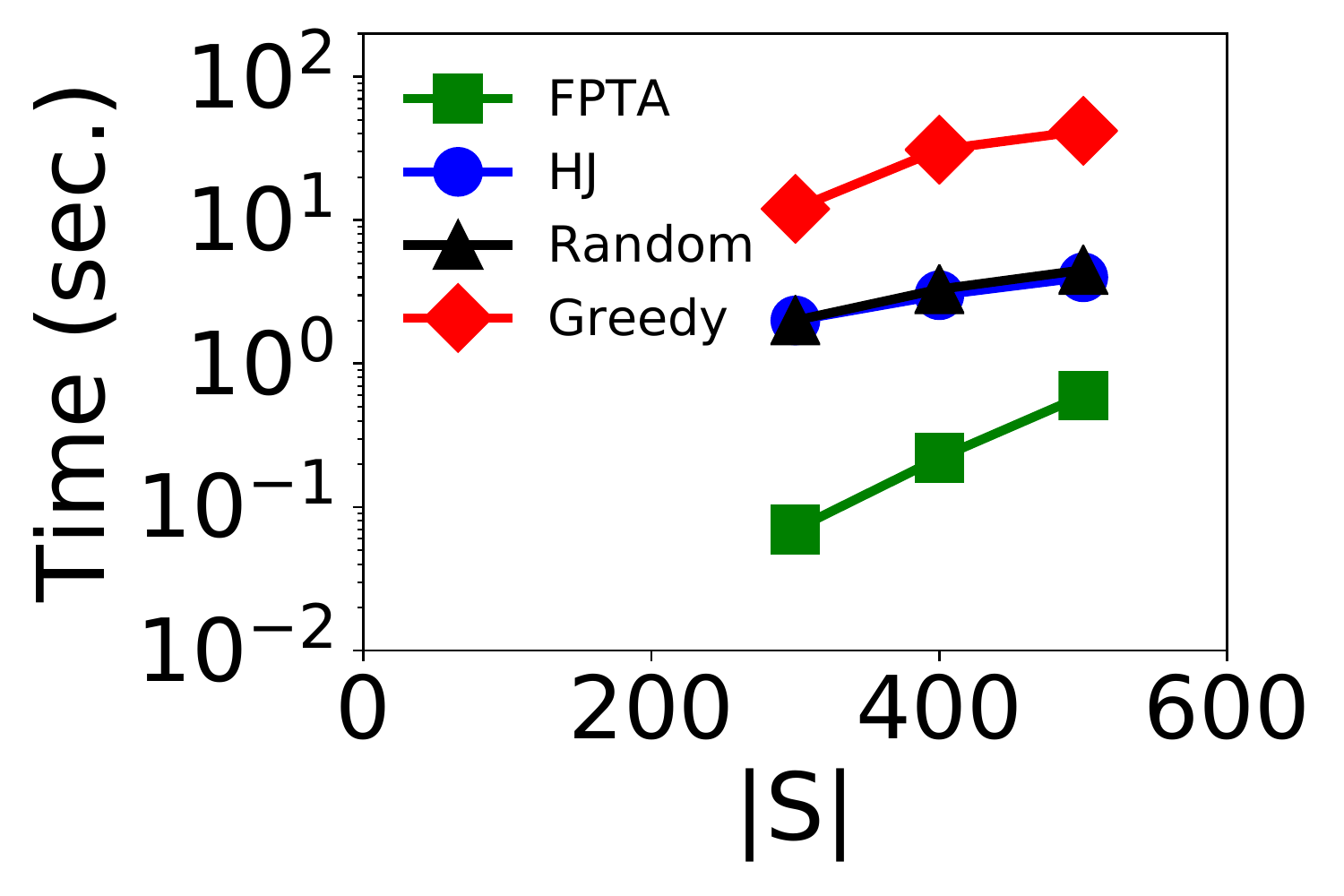}\label{fig:ER_time}}
     \subfloat[BA]{\includegraphics[width=0.4\textwidth]{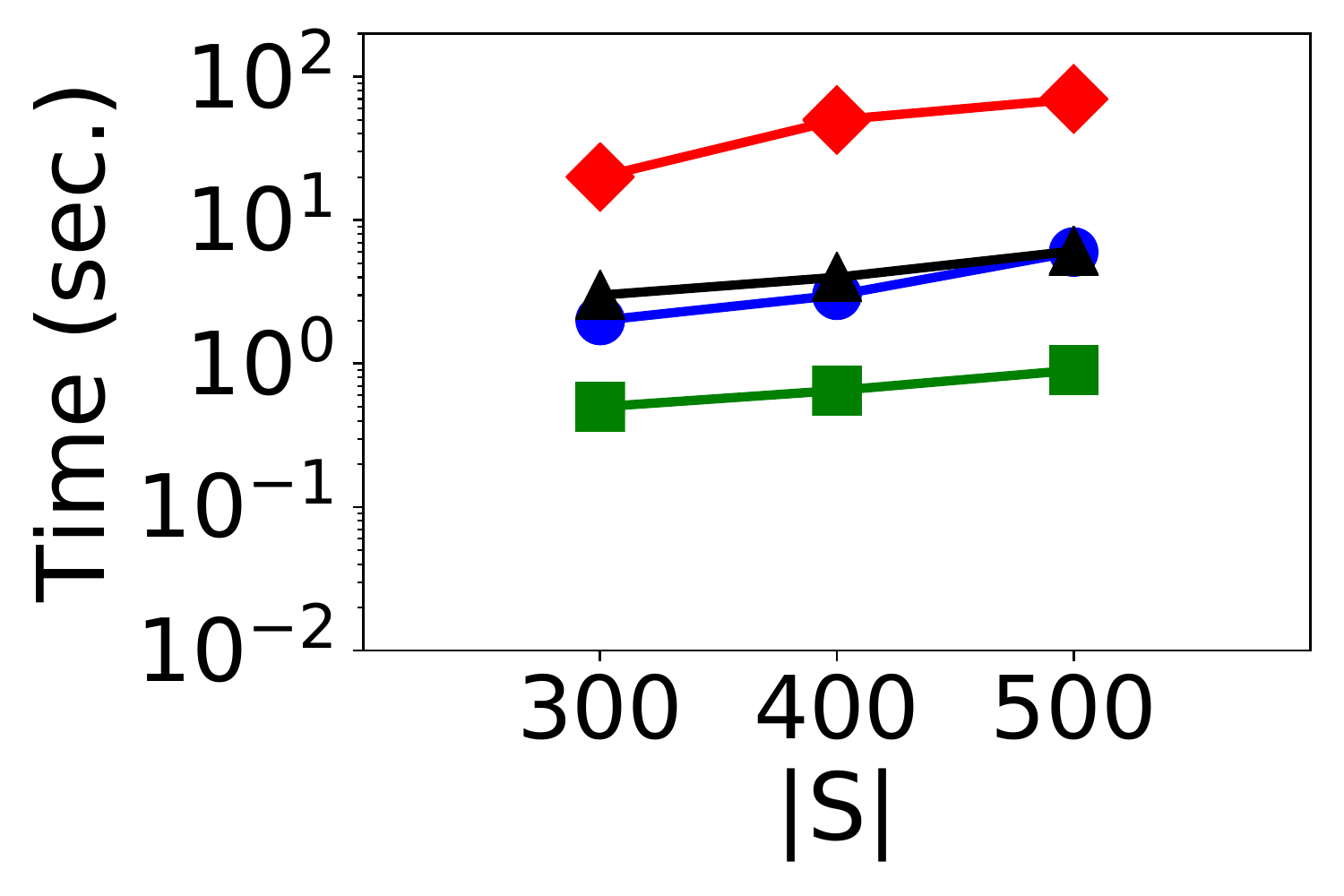}\label{fig:BA_time}} 
    \caption{ \textbf{Synthetic Graphs: }Time taken by different algorithms: Our algorithm FPTA outperforms other baselines in terms of efficiency. \label{fig:baselines_time_synthetic}}
\end{figure}

\begin{table}[t]
\centering
{\small
\begin{tabular}{| c|c |c | c|c |}
\hline
\textbf{Graph}& $|S|$ &  \textbf{FPTA} & \textbf{HJ} & \textbf{Random}\\
\hline
 \multirow{3}{*}{Social}&30 & 9 & 3k &4k \\\cline{2-5}
 & 50 & 14 & 9k &10k\\ \cline{2-5}
 & 70 & 20 & 13k &15k\\ \cline{2-5}
 \hline
 \multirow{3}{*}{Power}&200 &3 & 473 & 803 \\\cline{2-5}
 & 300 & 7 & 682& 1.3k\\ \cline{2-5}
 & 400 & 11 & 721&  1.5k\\ \cline{2-5}
 \hline
 \multirow{3}{*}{Web}& 200 & 3 & 1.9k& 2.3k \\\cline{2-5}
 & 300 & 7 & 4k &4.3k\\ \cline{2-5}
 & 400 & 12 & 5k & 6.3k\\ \cline{2-5}
 \hline
 \multirow{3}{*}{Road}&200 & 3 & 421 & 400 \\\cline{2-5}
 & 400 & 4 & 800 & 623\\ \cline{2-5}
 & 600 & 7 & 1.1k & 1.2k \\ \cline{2-5}
 
\hline
\end{tabular}
}
\caption{Comparison for the number of edges produced by FPTA and Random. \label{table:quality_random}}
 \end{table}

%% file: previous_work.tex
\section{Related Work} \label{sec:prev_work}
Zhou et al. initiated the study of the problem of attacking node based similarity measures via edge removal~\cite{Zhou:2019:ASL:3306127.3331707}. The authors proved that the problems based on local and global structural similarities are NP-hard and provided some heuristics. However, we focus on the parameterized complexity of this problem and its variations. Our work is also related to studying vulnerability of
social network analysis (SNA). Attacking against centrality measures via edge manipulation had been studied in the past \cite{waniek2017construction,dey2019covert}. These papers show that it is computationally hard to hide for a leader inside a covert network. Waniek et al. \cite{waniek2018attack} and Yu et al. \cite{yu2018target} recently discussed the problem of hiding or anonymizing links on networks. Similarly, measuring network robustness via network modification is also a well studied problem. Laishram et al. \cite{Laishram2018} formulated network resilience in terms of the stability of $k$-cores against deletion of random nodes/edges. Other work \cite{zhang2017finding,medya2019k} discussed a related problem which was to find $b$ vertices such that their deletion reduces the $k$-core maximally. Zhu et al. studied this problem via edge deletion~\cite{zhu2018k}. 

All the above problems are in the category of manipulating network structure to optimize for certain objective. There is another line of related work concerning network design. Paik et al. first initiated this line of work by studying several network design problems based on vertex upgradation to decrease the delays on adjacent edges~\cite{paik1995}. Since then these  problems have received a significant amount of research attention. Meyerson et al. designed algorithms for the minimization of shortest path distances~\cite{meyerson2009}. A series of recent work~\cite{lin2015,dilkina2011,medya2018making,medya2018noticeable} studied the problem of minimizing the shortest path distance by improving edge or node weights. Another related line of research work concern the problem of increasing the centrality of a node or a set of nodes by adding edges~\cite{crescenzi2015,ishakian2012framework,medya2018group,amelkin2019fighting}. Boosting or containing diffusion processes in networks were investigated under different diffusion models such as Linear Threshold model \cite{Khalil2014}, credit distribution model \cite{medya2019influence}, and Independent Cascade model \cite{kimura2008minimizing,chaoji2012recommendations}.
 

%% file: conclusion.tex
\section{Conclusion}

We have proposed three graph theoretic problems and argued that they capture the main computational challenge of attacking local similarity measures. We have studied these problems in parameterized complexity framework. We have considered the number of edges that we are allowed to delete and the number of targeted pairs of nodes as our parameters and shown either existence of \FPT algorithm or parameterized intractability for them. We have also exhibited polynomial time algorithm for the \ESED problem in an important special case. We finally establish effectiveness of our FPT algorithm for the \ESED problem in real and synthetic data sets. An important future direction of research is to study kernelization techniques for these problems. The authors are not aware of any kernelization based work in social network analysis and believe that kernels would be useful in practice.

%% file: main.bbl
\begin{thebibliography}{10}

\bibitem{DBLP:conf/atal/SabaterS02}
Jordi Sabater and Carles Sierra.
\newblock Reputation and social network analysis in multi-agent systems.
\newblock In {\em Proc. 1st International Joint Conference on Autonomous Agents
  and Multiagent Systems, {AAMAS}}, pages 475--482, 2002.

\bibitem{otte2002social}
Evelien Otte and Ronald Rousseau.
\newblock Social network analysis: a powerful strategy, also for the
  information sciences.
\newblock {\em J. Inf. Sci.}, 28(6):441--453, 2002.

\bibitem{wang2007social}
Fei-Yue Wang, Kathleen~M Carley, Daniel Zeng, and Wenji Mao.
\newblock Social computing: From social informatics to social intelligence.
\newblock {\em IEEE Intelligent systems}, 22(2), 2007.

\bibitem{carrington2005models}
Peter~J Carrington, John Scott, and Stanley Wasserman.
\newblock {\em Models and methods in social network analysis}, volume~28.
\newblock Cambridge university press, 2005.

\bibitem{al2006link}
Mohammad Al~Hasan, Vineet Chaoji, Saeed Salem, and Mohammed Zaki.
\newblock Link prediction using supervised learning.
\newblock In {\em SDM06: workshop on link analysis, counter-terrorism and
  security}, 2006.

\bibitem{liben2007link}
David Liben-Nowell and Jon Kleinberg.
\newblock The link-prediction problem for social networks.
\newblock {\em Journal of the American society for information science and
  technology}, 58(7):1019--1031, 2007.

\bibitem{wang2015link}
Peng Wang, BaoWen Xu, YuRong Wu, and XiaoYu Zhou.
\newblock Link prediction in social networks: the state-of-the-art.
\newblock {\em Science China Information Sciences}, 58(1):1--38, 2015.

\bibitem{zhou2009predicting}
Tao Zhou, Linyuan L{\"u}, and Yi-Cheng Zhang.
\newblock Predicting missing links via local information.
\newblock {\em The European Physical Journal B}, 71(4):623--630, 2009.

\bibitem{lim2019hidden}
Marcus Lim, Azween Abdullah, NZ~Jhanjhi, and Mahadevan Supramaniam.
\newblock Hidden link prediction in criminal networks using the deep
  reinforcement learning technique.
\newblock {\em Computers}, 8(1):8, 2019.

\bibitem{chen2005link}
Hsinchun Chen, Xin Li, and Zan Huang.
\newblock Link prediction approach to collaborative filtering.
\newblock In {\em Proceedings of the 5th ACM/IEEE-CS Joint Conference on
  Digital Libraries (JCDL'05)}, pages 141--142. IEEE, 2005.

\bibitem{talasu2017link}
Nitish Talasu, Annapurna Jonnalagadda, S~Sai~Akshaya Pillai, and Jampani Rahul.
\newblock A link prediction based approach for recommendation systems.
\newblock In {\em 2017 international conference on advances in computing,
  communications and informatics (ICACCI)}, pages 2059--2062. IEEE, 2017.

\bibitem{campana2017recommender}
Mattia~G Campana and Franca Delmastro.
\newblock Recommender systems for online and mobile social networks: A survey.
\newblock {\em Online Social Networks and Media}, 3:75--97, 2017.

\bibitem{Zhou:2019:ASL:3306127.3331707}
Kai Zhou, Tomasz~P. Michalak, Marcin Waniek, Talal Rahwan, and Yevgeniy
  Vorobeychik.
\newblock Attacking similarity-based link prediction in social networks.
\newblock In {\em Proc. 18th International Conference on Autonomous Agents and
  MultiAgent Systems}, pages 305--313, 2019.

\bibitem{CyganEtAl}
Marek Cygan, Fedor~V. Fomin, Lukasz Kowalik, Daniel Lokshtanov, D{\'{a}}niel
  Marx, Marcin Pilipczuk, Michal Pilipczuk, and Saket Saurabh.
\newblock {\em Parameterized Algorithms}.
\newblock Springer, 2015.

\bibitem{DBLP:conf/mfcs/ChenKX06}
Jianer Chen, Iyad~A. Kanj, and Ge~Xia.
\newblock Improved parameterized upper bounds for vertex cover.
\newblock In {\em Proc. 31st Mathematical Foundations of Computer Science 2006,
  31st International Symposium, {MFCS}}, pages 238--249, 2006.

\bibitem{DBLP:books/daglib/0004338}
Vijay~V. Vazirani.
\newblock {\em Approximation algorithms}.
\newblock Springer, 2001.

\bibitem{garey1974some}
Michael~R Garey, David~S Johnson, and Larry Stockmeyer.
\newblock Some simplified np-complete problems.
\newblock In {\em Proc. sixth annual ACM symposium on Theory of computing},
  pages 47--63. ACM, 1974.

\bibitem{barabasi1999emergence}
Albert-L{\'a}szl{\'o} Barab{\'a}si and R{\'e}ka Albert.
\newblock Emergence of scaling in random networks.
\newblock {\em Science}, 286(5439):509--512, 1999.

\bibitem{erdHos1960evolution}
Paul Erd{\H{o}}s and Alfr{\'e}d R{\'e}nyi.
\newblock On the evolution of random graphs.

\bibitem{medya2018noticeable}
Sourav Medya, Jithin Vachery, Sayan Ranu, and Ambuj Singh.
\newblock Noticeable network delay minimization via node upgrades.
\newblock {\em Proceedings of the VLDB Endowment}, 11(9):988--1001, 2018.

\bibitem{silva2015hierarchical}
Arlei Silva, Petko Bogdanov, and Ambuj~K Singh.
\newblock Hierarchical in-network attribute compression via importance
  sampling.
\newblock In {\em 2015 IEEE 31st International Conference on Data Engineering},
  pages 951--962. IEEE, 2015.

\bibitem{waniek2017construction}
Marcin Waniek, Tomasz~P Michalak, Talal Rahwan, and Michael Wooldridge.
\newblock On the construction of covert networks.
\newblock In {\em Proceedings of the 16th Conference on Autonomous Agents and
  MultiAgent Systems}, pages 1341--1349, 2017.

\bibitem{dey2019covert}
Palash Dey and Sourav Medya.
\newblock Covert networks: How hard is it to hide?
\newblock In {\em Proceedings of the 18th International Conference on
  Autonomous Agents and MultiAgent Systems}, pages 628--637, 2019.

\bibitem{waniek2018attack}
Marcin Waniek, Kai Zhou, Yevgeniy Vorobeychik, Esteban Moro, Tomasz~P Michalak,
  and Talal Rahwan.
\newblock Attack tolerance of link prediction algorithms: How to hide your
  relations in a social network.
\newblock {\em arXiv preprint arXiv:1809.00152}, 2018.

\bibitem{yu2018target}
Shanqing Yu, Minghao Zhao, Chenbo Fu, Huimin Huang, Xincheng Shu, Qi~Xuan, and
  Guanrong Chen.
\newblock Target defense against link-prediction-based attacks via evolutionary
  perturbations.
\newblock {\em arXiv preprint arXiv:1809.05912}, 2018.

\bibitem{Laishram2018}
Ricky Laishram, Ahmet~Erdem Sariy\"{u}ce, Tina Eliassi-Rad, Ali Pinar, and
  Sucheta Soundarajan.
\newblock Measuring and improving the core resilience of networks.
\newblock In {\em Proceedings of the 2018 World Wide Web Conference}, pages
  609--618, 2018.

\bibitem{zhang2017finding}
Fan Zhang, Ying Zhang, Lu~Qin, Wenjie Zhang, and Xuemin Lin.
\newblock Finding critical users for social network engagement: The collapsed
  k-core problem.
\newblock In {\em Thirty-First AAAI Conference on Artificial Intelligence},
  pages 245--251, 2017.

\bibitem{medya2019k}
Sourav Medya, Tiyani Ma, Arlei Silva, and Ambuj Singh.
\newblock K-core minimization: A game theoretic approach.
\newblock {\em arXiv preprint arXiv:1901.02166}, 2019.

\bibitem{zhu2018k}
Weijie Zhu, Chen Chen, Xiaoyang Wang, and Xuemin Lin.
\newblock K-core minimization: An edge manipulation approach.
\newblock In {\em Proceedings of the 27th ACM International Conference on
  Information and Knowledge Management}, pages 1667--1670. ACM, 2018.

\bibitem{paik1995}
D.~Paik and S.~Sahni.
\newblock Network upgrading problems.
\newblock {\em Networks}, 1995.

\bibitem{meyerson2009}
Adam Meyerson and Brian Tagiku.
\newblock Minimizing average shortest path distances via shortcut edge
  addition.
\newblock In {\em Approximation, Randomization, and Combinatorial Optimization.
  Algorithms and Techniques (APPROX-RANDOM)}, pages 272--285. Springer, 2009.

\bibitem{lin2015}
Yimin Lin and Kyriakos Mouratidis.
\newblock Best upgrade plans for single and multiple source-destination pairs.
\newblock {\em GeoInformatica}, 19(2):365--404, 2015.

\bibitem{dilkina2011}
Bistra Dilkina, Katherine~J. Lai, and Carla~P. Gomes.
\newblock Upgrading shortest paths in networks.
\newblock In {\em Integration of AI and OR Techniques in Constraint Programming
  for Combinatorial Optimization Problems}, pages 76--91. Springer, 2011.

\bibitem{medya2018making}
Sourav Medya, Petko Bogdanov, and Ambuj Singh.
\newblock Making a small world smaller: Path optimization in networks.
\newblock {\em IEEE Transactions on Knowledge and Data Engineering},
  30(8):1533--1546, 2018.

\bibitem{crescenzi2015}
Pierluigi Crescenzi, Gianlorenzo D'Angelo, Lorenzo Severini, and Yllka Velaj.
\newblock Greedily improving our own centrality in a network.
\newblock In {\em SEA}, pages 43--55. Springer International Publishing, 2015.

\bibitem{ishakian2012framework}
Vatche Ishakian, D{\'o}ra Erdos, Evimaria Terzi, and Azer Bestavros.
\newblock A framework for the evaluation and management of network centrality.
\newblock In {\em Proc. SIAM International Conference on Data Mining}, pages
  427--438, 2012.

\bibitem{medya2018group}
Sourav Medya, Arlei Silva, Ambuj Singh, Prithwish Basu, and Ananthram Swami.
\newblock Group centrality maximization via network design.
\newblock In {\em Proc. 24th SIAM International Conference on Data Mining},
  pages 126--134. SIAM, 2018.

\bibitem{amelkin2019fighting}
Victor Amelkin and Ambuj~K Singh.
\newblock Fighting opinion control in social networks via link recommendation.
\newblock In {\em Proceedings of the 25th ACM SIGKDD International Conference
  on Knowledge Discovery \& Data Mining}, pages 677--685, 2019.

\bibitem{Khalil2014}
Elias~Boutros Khalil, Bistra Dilkina, and Le~Song.
\newblock Scalable diffusion-aware optimization of network topology.
\newblock In {\em SIGKDD international conference on Knowledge discovery and
  data mining}, pages 1226--1235. ACM, 2014.

\bibitem{medya2019influence}
Sourav Medya, Arlei Silva, and Ambuj Singh.
\newblock Influence minimization under budget and matroid constraints: Extended
  version.
\newblock {\em arXiv preprint arXiv:1901.02156}, 2019.

\bibitem{kimura2008minimizing}
Masahiro Kimura, Kazumi Saito, and Hiroshi Motoda.
\newblock Minimizing the spread of contamination by blocking links in a
  network.
\newblock In {\em AAAI}, 2008.

\bibitem{chaoji2012recommendations}
Vineet Chaoji, Sayan Ranu, Rajeev Rastogi, and Rushi Bhatt.
\newblock Recommendations to boost content spread in social networks.
\newblock In {\em WWW}, pages 529--538, 2012.

\end{thebibliography}
